\documentclass[a4paper,10pt]{article}
\usepackage[utf8]{inputenc}

\usepackage{amsmath,amsfonts,amssymb}

\newtheorem{theorem}{Theorem}

\newtheorem{proposition}[theorem]{Proposition}
\newenvironment{proof}[1][Proof]{\noindent\textbf{#1: }}{\ \rule{0.5em}{0.5em}}

\title{The geometric lattice of embedded subsets}
\author{Giovanni Rossi\\
\footnotesize{Department of Computer Science and Engineering DISI, University of Bologna}\\
\footnotesize{Mura Anteo Zamboni 7, Italy 40126 - email: giovanni.rossi6@unibo.it}}

\begin{document}

\maketitle

\begin{abstract}
This work proposes an alternative approach to the so-called lattice of embedded subsets, which is included in the product of the subset and partition lattices of a finite set,
and whose elements are pairs consisting of a subset and a partition where the former is a block of the latter. The lattice structure proposed in a recent contribution relies
on ad-hoc definitions of both the join operator and the bottom element, while also including join-irreducible elements distinct from atoms. Conversely, here embedded subsets
obtain through a closure operator defined over the product of the subset and partition lattices, where elements are generic pairs of a subset and a partition. Those such pairs
that coincide with their closure are precisely embedded subsets, and since the Steinitz exchange axiom is also satisfied, what results is a geometric (hence atomic) lattice
given by a simple matroid (or combinatorial geometry) included in the product of the subset and partition lattices (as the partition lattice itself is the polygon matroid
defined on the edges of a complete graph). By focusing on its M\"obius function, this geometric lattice of embedded subsets of a $n$-set is shown to be isomorphic to the
lattice of partitions of a $n+1$-set.\smallskip

\smallskip

\end{abstract}

\section{Introduction}
For a finite set $N=\{1,\ldots ,n\}$, the framework developed in the sequel obtains from both the lattice $(2^N,\cap,\cup)$ of subsets of $N$ and the lattice
$(\mathcal P^N,\wedge,\vee)$ of partitions of $N$. Generic subsets and partitions are denoted, respectively, by $A,B\in 2^N$ and $P,Q\in\mathcal P^N$. Recall that a partition
$P=\{A_1,\ldots ,A_{|P|}\}$ is a family of (non-empty) pair-wise disjoint subsets $A_1,\ldots ,A_{|P|}\in 2^N$, called blocks, whose union is $N$. As usual
\cite{Aigner79,Stern99}, partitions are (partially) ordered by coarsening $\geqslant$, meaning that any $P,Q\in\mathcal P^N$ satisfy $P\geqslant Q$ if every block $B\in Q$
is included in a block $A\in P$, i.e. $A\supseteq B$. Hence the bottom partition is $P_{\bot}=\{\{1\},\ldots ,\{n\}\}$, while the top one is $P^{\top}=\{N\}$.

Denote by $X^N=2^N\times\mathcal P^N$ the product of the subset and partition lattices, whose elements $(A,P),(B,Q)\in X^N$ are pairs consisting of a subset $A,B\in 2^N$ and a
partition $P,Q\in\mathcal P^N$. Lattice $(X^N,\sqcap,\sqcup)$ is partially ordered by $\sqsupseteq$ and naturally obtains through product of the meet, join and order relations
for the subset and partition lattices. That is to say, $(A,P)\sqcap(B,Q)=(A\cap B,P\wedge Q)$ as well as $(A,P)\sqcup(B,Q)=(A\cup B,P\vee Q)$, while
$(A,P)\sqsupseteq(B,Q)\Leftrightarrow A\supseteq B,P\geqslant Q$, where $P\wedge Q$ is the coarsest partition finer than both $P,Q$ and $P\vee Q$ is the finest partition coarser
than both $P,Q$. Thus the bottom element is $(\emptyset,P_{\bot})$ while the top one is $(N,P^{\top})$.

The existing lattice $(\mathfrak C(N)_{\bot},\sqcap,\sqcup_*)$ of embedded subsets has elements given by those pairs $(A,P)\in X^N$ such that $A\in P$, together with ``\textit{an
artificial bottom element $\perp$, which is introduced for mathematical convenience and that may be considered as $(\emptyset,P_{\bot})$}'' \cite[p. 481]{Grabisch2010}. It is not
hard to see that if $A,B\in 2^N$ and $P,Q\in\mathcal P^N$ satisfy $A\in P$ as well as $B\in Q$, then non-emptiness $A\cap B\neq\emptyset$ entails
$(A\cap B,P\wedge Q)\in\mathfrak C(N)_{\bot}$, since $A\cap B$ is a block of $P\wedge Q$. However, in general $A\cup B$ is \textit{not} a block of $P\vee Q$, i.e.
$(A\cup B,P\vee Q)\notin\mathfrak C(N)_{\bot}$. For this reason, the existing lattice of embedded subsets $\mathfrak C(N)_{\bot}$ has the same meet $\sqcap$ and order $\sqsupseteq$
applying to product lattice $X^N$, but the join $\sqcup_*$ is different from $\sqcup$ as detailed below. The resulting structure \cite[Prop. 2, p. 481]{Grabisch2010} has $n$
atoms of the form $(i,P_{\bot}),i\in N$ together with additional $\binom{n}{2}(n-2)$ join-irreducible elements \cite[p. 32]{Aigner79}, and thus the lattice is non-atomic (hence
non-geometric) precisely because many elements cannot be decomposed as a join of atoms. In fact, the only embedded subsets $(A,P)\in\mathfrak C(N)_{\bot}$ admitting such a
decomposition \cite[Prop. 8, pp. 485-486]{Grabisch2010} are those where partition $P$ is the modular element \cite{Aigner79,Stanley1971} of the partition lattice whose
unique non-singleton block, if any, is precisely $A$ \cite[p. 71]{Aigner79} \cite[p. 339]{Stanley2012EnuCom}. This is denoted by $P=P^A_{\bot}$, with
$P^A_{\bot}=\{A,\{i_1\},\ldots ,\{i_{n-|A|}\}\}$ and $\{i_1,\ldots ,i_{n-|A|}\}=N\backslash A=A^c$.

Since $P^i_{\bot}=P_{\bot}$ for all $i\in N$ and $P^N_{\bot}=P^{\top}$, there are $2^n-n$ such modular partitions $P^A_{\bot},\emptyset\neq A\in 2^N$. They play a key role in the
approach proposed here, which relies on a closure operator $cl:X^N\rightarrow X^N$ over poset $(X^N,\sqsupseteq)$, meaning \cite[p. 167]{Aigner79} that for all
$(A,P),(B,Q)\in X^N$ the following hold:
\begin{itemize}
\item[(i)] $cl(A,P)\sqsupseteq(A,P)$,
\item[(ii)] $(A,P)\sqsupseteq (B,Q)\Rightarrow cl(A,P)\sqsupseteq cl(B,Q)$,
\item[(iii)] $cl(cl(A,P))=cl(A,P)$.
\end{itemize}
This closure is defined by $cl(\emptyset,P)=(\emptyset,P)$ for all $P\in\mathcal P^N$, while for $\emptyset\neq A$
\begin{equation}
cl(A,P)=(\bar A,P\vee P^A_{\bot})
\end{equation}
such that $A\subseteq\bar A\in P\vee P^A_{\bot}$. In words, $\bar A$ is the block of $P\vee P^A_{\bot}$ satisfying $\bar A\supseteq A$, and $P^A_{\bot}$ is the modular partition
where the only non-singleton block, if any, is $A$.  The remainder of this work is concerned with poset $(\mathcal E^N,\sqsupseteq)$, with $\mathcal E^N=\{(A,P):cl(A,P)=(A,P)\}$
consisting of those pairs that coincide with their closure. As $\perp$ corresponds to $(\emptyset,P_{\bot})$ (see above), $\mathfrak C(N)_{\bot}\subset\mathcal E^N\subset X^N$,
and in particular $\mathcal E^N\backslash\mathfrak C(N)_{\bot}=\{(\emptyset,P):P\neq P_{\bot}\}$. From the enumerative perspective,
$|\mathfrak C(N)_{\bot}|=1+\sum_{1\leq k\leq n}k\mathcal S_{n,k}$, where $\mathcal S_{n,k}$ is the Stirling number of the second kind, i.e. the number of partitions of a $n$-set
into $k$ blocks, while $|\mathcal E^N|=\sum_{1\leq k\leq n}(k+1)\mathcal S_{n,k}=\mathcal B_{n+1}$, where $\mathcal B_{n+1}$ is the Bell number, i.e. the number of partitions of
a $n+1$-set \cite{Aigner79,ConcreteMathematics,Rota1964}.

The join obtained through closure $cl(\cdot)$ coincides, over $\mathfrak C(N)_{\bot}$, with the existing one $\sqcup_*$ (defined in \cite[Prop. 2, p. 481]{Grabisch2010}), i.e. for
$(A,P),(B,Q)\in\mathfrak C(N)_{\bot}$,
\begin{equation*}
cl((A,P)\sqcup(B,Q))=(A,P)\sqcup_*(B,Q).
\end{equation*}
However, the novel enlarged lattice $\mathcal E^N$ of embedded subsets is substantially different from the existing one $\mathfrak C(N)_{\bot}$, as all its
$n+\binom{n}{2}=\binom{n+1}{2}$ join-irreducible elements (apart from $(\emptyset,P_{\bot})$) are atoms. Hence $\mathcal E^N$ is a geometric lattice with the same atoms as product
lattice $X^N$, and in fact the former reproduces within the latter a situation similar to the inclusion of partition lattice $(\mathcal P^N,\wedge,\vee)$ within Boolean lattice
$(2^{N_2},\cap,\cup)$, where $N_2=\{\{i,j\}:1\leq i<j\leq n\}$ is the $\binom{n}{2}$-set whose elements are all $A\in 2^N$ such that $|A|=2$. Since $N_2$ is the edge set of the
complete graph $K_N=(N,N_2)$ on vertex set $N$, partition lattice $(\mathcal P^N,\wedge,\vee)$ obtains through the closure $cl:2^{N_2}\rightarrow 2^{N_2}$ such that for every edge
set $E\in 2^{N_2}$ graph $G'=(N,cl(E))$ obtains from graph $G=(N,E)$ by adding all edges within each component, i.e. each component of $G'$ is a clique (or maximal complete subgraph)
\cite[p. 54]{Aigner79}. Thus partition lattice $\mathcal P^N$ corresponds to the family $\{E:E=cl(E)\}\subset 2^{N_2}$ of edge sets that coincide with their closure, with the
resulting structure known as the polygon matroid defined on the edges of the complete graph on vertex set $N$ \cite[pp. 259, 274]{Aigner79}. In the sequel, closure
$cl:X^N\rightarrow X^N$ identifying the novel lattice $\mathcal E^N$ of embedded subsets is shown to provide a simple matroid (or combinatorial geometry) \cite[pp. 52, 256]{Aigner79}.

The following Section 2 details how the closure defined by expression (1) above satisfies the Steinitz exchange axiom, thereby showing that the resulting lattice of embedded subsets
is geometric. Section 3 identifies the modular elements of this lattice for a $n$-set, while also highlighting the isomorphism with the lattice of partitions of a $n+1$-set.
Section 4 is devoted to M\"obius inversion, determining the M\"obius function of the geometric lattice of embedded subsets of a $n$-set, which equals the M\"obius function
of the lattice of partitions of a $n+1$-set. Since functions taking real values on embedded subsets were firstly considered in cooperative game theory (as ``games in partition
function form''), Section 5 focuses on the (free) vector space of such functions and on the associated M\"obius algebra, with emphasis on the vector subspaces identified by those
lattice functions whose M\"obius inversion lives only on modular elements. Section 6 concludes the paper with some final remarks.

\section{Partitions, atoms and closure}
The rank $r:\mathcal P^N\rightarrow\mathbb Z_+$ of partitions is $r(P)=n-|P|$, and atoms are immediately above the bottom element, with rank 1 \cite[p. 52]{Aigner79}. Hence they
consist of $n-1$ blocks, out of which $n-2$ are singletons while the remaining one is a pair. As already outlined, this means that $\mathcal P^N$ has essentially the same
$\binom{n}{2}$ atoms as $2^{N_2}$.  For $1\leq i<j\leq n$, denote by $[ij]$ the atom of $\mathcal P^N$ whose non-singleton block is pair $\{i,j\}$. A fundamental difference between
Boolean lattice $(2^{N_2},\cap ,\cup)$ and partition lattice $(\mathcal P^N,\wedge ,\vee)$ is in terms of join-decompositions \cite[Ch. II]{Aigner79}. In fact, every
$\{\{i,j\}_1,\ldots ,\{i,j\}_{|E|}\}=E\in 2^{N_2}$ admits a unique decomposition as a join of atoms, namely $E=\{i,j\}_1\cup\cdots\cup\{i,j\}_{|E|}$, which is irredundant.
Conversely, a generic partition $P\in\mathcal P^N$ admits several such decompositions $P=[ij]_1\vee\cdots\vee[ij]_k$ (with $k\geq r(P)$), most of which redundant, while the unique
maximal one (in terms of inclusion) clearly is the join of all atoms $[ij]$ such that $[ij]\leqslant P$. In particular, $r(P)$ is the minimum number of atoms in a join-decomposition
of $P$, thus any $P=[ij]_1\vee\cdots\vee[ij]_{r(P)}$ is irredundant. 

The classical examples of semimodular lattices \cite{Stern99} come from sets endowed with a closure operator satisfying the Steinitz exchange axiom. Specifically, set $N_2$
endowed with the closure operator $cl:2^{N_2}\rightarrow 2^{N_2}$ introduced in Section 1 yields the semimodular lattice $\mathcal P^N$ of partitions. Subsets $E\in 2^{N_2}$ such
that $cl(E)=E$ are said to be closed, and the family of closed subsets forms a complete lattice under inclusion with meet given by intersection and join given by the closure of
the union. That is, if $E,E'$ are closed, then $E\cap E'$ is closed as well, while $cl(E\cup E')$ is the smallest closed subset containing both $E$ and $E'$. Furthermore, for all
$E\in 2^{N_2}$ and $\{i,j\},\{i,j\}'\in N_2$, the closure also satisfies the following Steinitz exchange axiom:
\begin{equation*}
\text{if }\{i,j\}\notin cl(E)\text{ and }\{i,j\}\in cl(E\cup\{i,j\}')\text{ then }\{i,j\}'\in cl(E\cup\{i,j\})\text .
\end{equation*}
Hence the family of closed subsets is a matroid (included in $2^{N_2}$) which, in particular, is said to be simple (or a combinatorial geometry) in that $cl(\emptyset)=\emptyset$
and $cl(\{i,j\})=\{i,j\}$ for all $\{i,j\}\in N_2$ \cite[pp. 52-53]{Aigner79}.

As already outlined, a simple matroid also obtains as the novel enlarged lattice $\mathcal E^N$ of embedded subsets, whose elements are closed pairs, i.e. pairs $(A,P)\in X^N$
that coincide with their closure $cl(A,P)$ as defined by expression (1) above. To see this, firstly consider the covering relation between closed pairs, denoted by $\sqsupset^*$,
i.e. for $(A,P),(B,Q)\in\mathcal E^N$, it holds $(A,P)\sqsupset^*(B,Q)$ whenever $\{(B',Q'):(A,P)\sqsupseteq (B',Q')\sqsupseteq (B,Q),(B',Q')\in\mathcal E^N\}=\{(A,P),(B,Q)\}$
(see \cite[p. 3]{Aigner79}). The general situation is when either $A\neq\emptyset\neq B$ or else $A=\emptyset=B$, in which case $(A,P)\sqsupset^*(B,Q)$ attains whenever
$P\gtrdot Q$, where $\gtrdot$ denotes the covering relation between partitions, i.e. $P$ obtains by merging two blocks of $Q$ into a unique block. In addition, it can also be
$A\neq\emptyset=B$, and in this case $(A,P)\sqsupset^*(B,Q)$ attains for $P=Q$. This means that $\mathcal E^N$ has $n+\binom{n}{2}=\binom{n+1}{2}$ atoms, in that the
bottom pair $(\emptyset,P_{\bot})$ is covered by $n$ embedded subsets $(i,P_{\bot})$ for $i\in N$, and by further $\binom{n}{2}$ embedded subsets $(\emptyset,[ij])$ for
$\{i,j\}\in N_2$.

Let $\mathcal A_1=\{(i,P_{\bot}):i\in N\}$ and $\mathcal A_2=\{(\emptyset,[ij]):\{i,j\}\in N_2\}$. Hence the set of atoms of $\mathcal E^N$ is
$\mathcal A=\mathcal A_1\cup\mathcal A_2$. The closure operator defined by expression (1) above can now be regarded as a mapping
$cl:2^{\mathcal A}\rightarrow 2^{\mathcal A}$, with closed subsets of $\mathcal A$ being the elements of lattice $\mathcal E^N$. In this view, consider a non-closed subset
$S\in\mathcal A$, with $S=S_1\cup S_2,S_1\in2^{\mathcal A_1},S_2\in2^{\mathcal A_2}$. Firstly focus on the simple case where $S_1=\emptyset$, thus
$S=S_2=\{(\emptyset,[ij]_1),\ldots,(\emptyset,[ij]_{|S_2|})\}$. Here the closure $cl(S)$ of $S$ obtains by adding all atoms $(\emptyset,[ij])\in\mathcal A_2$ such that $[ij]$ is
finer than $[ij]_1\vee\cdots\vee[ij]_{|S_2|}$, i.e. $cl(S)=S\cup\{(\emptyset,[ij]):[ij]\leqslant[ij]_1\vee\cdots\vee[ij]_{|S_2|}\}$. Hence for $S_1=\emptyset$ the closure
$cl:2^{\mathcal A}\rightarrow 2^{\mathcal A}$ applying to embedded subsets reproduces the closure $cl:2^{N_2}\rightarrow 2^{N_2}$ applying to partitions
described above. For the general case $\{(i_1,P_{\bot}),\ldots ,(i_{|S_1|},P_{\bot})\}=S_1\neq\emptyset\neq S_2=\{(\emptyset,[ij]_1),\ldots,(\emptyset,[ij]_{|S_2|})\}$, let
$\{i_1,\ldots ,i_{|S_1|}\}=A$ and $[ij]_1\vee\cdots\vee[ij]_{|S_2|}=P$, with $(A,P)\in X^N$. In view of the closure applying to partitions, like in the previous case, if
$[ij]\leqslant P$, then $(\emptyset,[ij])\in cl(S)$. At this point, the closure applying to embedded subsets defined in Section 1 works as follows. If $A\in P$ (entailing
$P=P\vee P^A_{\bot}$ or, equivalently, $P\geqslant P^A_{\bot}$, of course), then the only atoms to be added are those just described, i.e.
$cl(S)=S\cup\{(\emptyset,[ij]):[ij]\leqslant P\}$, and $(A,P)\in\mathcal E^N$. Otherwise, for $k\leq|P|$, let $B_1,\cdots,B_k\in P$ be \textit{all} the blocks of $P$ whose
intersection with $A$ is non-empty. That is, $\{B:B\in P,B\cap A\neq\emptyset\}=\{B_1,\ldots ,B_k\}$. Then, $\bar A=B_1\cup\cdots\cup B_k$ is a block of partition
$P\vee P^A_{\bot}$, with proper inclusion $\bar A\supset A$, and therefore $(\bar A,P\vee P^A_{\bot})\in\mathcal E^N$ is the minimal (in terms of $\sqsupseteq$) embedded subset
$(B,Q)$ satisfying $(B,Q)\sqsupseteq(A,P)$. On the other hand, the smallest closed subset of $\mathcal A$ containing $S$ is
$cl(S)=\{(i,P_{\bot}):i\in\bar A\}\cup\{(\emptyset,[ij]):[ij]\leqslant P\vee P^A_{\bot}\}$.

\begin{proposition}
Closure $cl:2^{\mathcal A}\rightarrow 2^{\mathcal A}$ satisfies the Steinitz exchange axiom.
\end{proposition}

\begin{proof}
For $(A,P)\in\mathcal E^N$, let atoms $(i,P_{\bot}),(i',P_{\bot})\in\mathcal A_1,(\emptyset,[ij]),(\emptyset,[ij]')\in\mathcal A_2$ satisfy 
$(A,P)\not\sqsupseteq(i,P_{\bot}),(i',P_{\bot}),(\emptyset,[ij]),(\emptyset,[ij]')$, i.e. $i,i'\not\in A$ and
$[ij],[ij]'\not\leqslant P$. In terms of subsets of atoms, let $S_{AP}=\{(\tilde i,P_{\bot}):\tilde i\in A\}\cup\{(\emptyset,\tilde{[ij]}):\tilde{[ij]}\leqslant P\}$, hence
$S_{AP}=cl(S_{AP})$ and $(i,P_{\bot}),(i',P_{\bot}),(\emptyset,[ij]),(\emptyset,[ij]')\notin S_{AP}$. Two cases
$(i',P_{\bot}),(\emptyset,[ij]')\in cl(S_{AP}\cup(i,P_{\bot}))$ and $(i',P_{\bot}),(\emptyset,[ij]')\in cl(S_{AP}\cup(\emptyset,[ij]))$ must be distinguished. In the first one,
there is a block $A'\in P,A'\neq A$ such that $i\in A'$, and both $i'\in A'$ and $|A\cap\{i,j\}'|=1=|A'\cap\{i,j\}'|$ hold. Therefore,
$cl(S_{AP}\cup(i',P_{\bot}))=cl(S_{AP}\cup(i,P_{\bot}))=cl(S_{AP}\cup(\emptyset,[ij]'))$, entailing that the exchange axiom is satisfied. Analogously, in the second case there is
a block $A'\in P,A'\neq A$ such that $|A\cap\{i,j\}|=1=|A'\cap\{i,j\}|$, and both $i'\in A'$ and $|A\cap\{i,j\}'|=1=|A'\cap\{i,j\}'|$ hold. Hence again the exchange axiom is
satisfied as $cl(S_{AP}\cup(i',P_{\bot}))=cl(S_{AP}\cup(\emptyset,[ij]))=cl(S_{AP}\cup(\emptyset,[ij]'))$.
\end{proof}
\smallskip

Hence, $\mathcal E^N$ is a matroid and, in particular, a simple one, in that all $\binom{n+1}{2}+1$ subsets of atoms given by the empty set and the $\binom{n+1}{2}$ singletons
are closed.

\begin{proposition}
Poset $(\mathcal E^N,\sqsupseteq)$ is a geometric lattice where the meet and join of any $(A,P),(B,Q)\in\mathcal E^N$ are $(A,P)\sqcap(B,Q)$ and $cl((A,P)\sqcup(B,Q))$,
respectively.
\end{proposition}

\begin{proof}
Since $|\mathcal E^N|<\infty$, it only has to be verified \cite[Prop. 2.28, p. 52]{Aigner79} that for all $(A,P),(B,Q)\in\mathcal E^N$ the following hold:\\
- if $(A,P)\sqsupset^*(B,Q)$, then there is an atom $(i',P_{\bot})\in\mathcal A_1,(B,Q)\not\sqsupseteq(i',P_{\bot})$ and/or
$(\emptyset,[ij])\in\mathcal A_2,(B,Q)\not\sqsupseteq(\emptyset,[ij])$ such that $(A,B)=cl((B,Q)\sqcup(i',P_{\bot}))$ and/or $(A,B)=cl((B,Q)\sqcup(\emptyset,[ij]))$,\\
- if $(A,B)=cl((B,Q)\sqcup(i',P_{\bot}))$ and/or $(A,B)=cl((B,Q)\sqcup(\emptyset,[ij]))$ for some atom $(i',P_{\bot})\in\mathcal A_1,(B,Q)\not\sqsupseteq(i',P_{\bot})$ and/or
$(\emptyset,[ij])\in\mathcal A_2,(B,Q)\not\sqsupseteq(\emptyset,[ij])$, then $(A,P)\sqsupset^*(B,Q)$.

Both conditions are trivially satisfied if $(B,Q)$ is the bottom element, hence let $(B,Q)\neq(\emptyset,P_{\bot})$. Then $(A,P)\sqsupset^*(B,Q)$ attains for
$P\gtrdot Q,A\supseteq B$. If $A=B$, then $(A,P)=cl((B,Q)\sqcup(\emptyset,[ij]))$ for all atoms $(\emptyset,[ij])\in\mathcal A_2,(B,Q)\not\sqsupseteq(\emptyset,[ij])$ such that
$\{i,j\}\cap B=\emptyset$ (where of course $Q\not\geqslant[ij]$), while proper inclusion $A\supset B$ entails that for some block $B'\in Q,B'\neq B$ the covering
embedded subset $(A,P)$ obtains as $(A,P)=cl((B,Q)\sqcup(\emptyset,[ij]))=cl((B,Q)\sqcup(i',P_{\bot}))$ for all atoms
$(i',P_{\bot})\in\mathcal A_1,(B,Q)\not\sqsupseteq(i',P_{\bot}),(\emptyset,[ij])\in\mathcal A_2,(B,Q)\not\sqsupseteq(\emptyset,[ij])$ such that $i'\in B'$ and
$|B\cap\{i,j\}|=1=|B'\cap\{i,j\}|$. Accordingly, $A=B\cup B'$.

Turning to the second condition, for all
$(i',P_{\bot})\in\mathcal A_1,(B,Q)\not\sqsupseteq(i',P_{\bot})$, $(\emptyset,[ij])\in\mathcal A_2,(B,Q)\not\sqsupseteq(\emptyset,[ij])$ it is easily checked that both
$cl((B,Q)\sqcup(i',P_{\bot}))$ and $cl((B,Q)\sqcup(\emptyset,[ij]))$ cover $(B,Q)$. In fact, the general case is where there is a block $B'\in Q,B'\neq B$ such that $i'\in B'$ and
$|B'\cap\{i,j\}|=1=|B\cap\{i,j\}|$, entailing $cl((B,Q)\sqcup(i',P_{\bot}))=cl((B,Q)\sqcup(\emptyset,[ij]))=(B\cup B',Q\vee P^{B\cup B'}_{\bot})$, with
$Q\vee P^{B\cup B'}_{\bot}$ obtained by merging blocks $B$ and $B'$ in $Q$, hence $Q\vee P^{B\cup B'}_{\bot}\gtrdot Q$. Yet, like for the first condition,
if $B\cap\{i,j\}=\emptyset$, then there are blocks $B',B''\in Q$, with $B'\neq B\neq B''$, such that $|B'\cap\{i,j\}|=1=|B''\cap\{i,j\}|$. Accordingly,
$cl((B,Q)\sqcup(\emptyset,[ij]))=(B,Q')$, where $Q'$ obtains by merging blocks $B',B''$ in $Q$, i.e. $Q'\gtrdot Q$.
\end{proof}
\smallskip

Having recognized the lattice structure characterizing $\mathcal E^N$, it may now be readily seen that the Jordan-Dedekind chain condition \cite[p. 30]{Aigner79} is satisfied, and
the rank function $r:\mathcal E^N\rightarrow\mathbb Z_+$ is $r(A,P)=r(P)+\min\{|A|,1\}$. In fact, for any $(A,P),(B,Q)\in\mathcal E^N$ such that $(A,P)\sqsupset(B,Q)$ (that is,
$(A,P)\sqsupseteq(B,Q)$, $(A,P)\neq(B,Q)$), all maximal chains $\{(B',Q')_0,(B',Q')_1,\ldots,(B',Q')_k\}$ from $(B,Q)=(B',Q')_0$ to $(A,P)=(B',Q')_k$
(i.e. $(B',Q')_{k'+1}\sqsupset^*(B',Q')_{k'}$ for $0\leq k'<k$) have same length $k$, in that $(B',Q')_{k'+1}=cl((B',Q')_{k'}\sqcup(i,P_{\bot}))$ and/or
$(B',Q')_{k'+1}=cl((B',Q')_{k'}\sqcup(\emptyset,[ij]))$ with atoms $(i,P_{\bot}),(\emptyset,[ij])$ such that
$(B',Q')_{k'+1}\sqsupseteq(i,P_{\bot}),(\emptyset,[ij])$ while $(B',Q')_{k'}\not\sqsupseteq(i,P_{\bot}),(\emptyset,[ij])$.

\section{Modular elements and isomorphism}
By definition (of modular elements of geometric lattices \cite[p. 58]{Aigner79}), modular embedded subsets are those $(A,P)\in\mathcal E^N$ such that for all
$(B,Q)\in\mathcal E^N$ it holds $r(A,P)+r(B,Q)=r(cl((A,P)\sqcup(B,Q)))+r((A,P)\sqcap(B,Q))$. For $Q\in\mathcal P^N$ and $\emptyset\neq A\in 2^N$, let
$Q^A=\{B\cap A:B\in Q,B\cap A\neq\emptyset\}$ denote the partition of $A$ induced by $Q$. 

\begin{proposition}
An embedded subset $(A,P)\in\mathcal E^N$ is modular if and only if $P$ is a modular partition and either $A=\emptyset$, or else $A$ is the non-singleton block of $P$ when
$P\neq P_{\bot}$.
\end{proposition}

\begin{proof}
For the \textquoteleft if\textquoteright\text{ }part, the bottom element $(\emptyset,P_{\bot})$ is evidently modular. Accordingly, firstly consider embedded subset
$(\emptyset,P^A_{\bot})$, with $|A|>1$. The rank is $r(\emptyset,P^A_{\bot})=n-(n-|A|+1)=|A|-1$. On the other hand, any $(B,Q)$ has rank $n-|Q|+1$ (since the case $B=\emptyset$
is trivial precisely because $P^A_{\bot}$ is a modular partition). Thus $r(\emptyset,P^A_{\bot})+r(B,Q)=n-|Q|+|A|$. Now consider that the meet has rank
$r((\emptyset,P^A_{\bot})\sqcap(B,Q))=n-(|Q^A|+n-|A|)=|A|-|Q^A|$, while the (closure of the) join yields
$r(cl((\emptyset,P^A_{\bot})\sqcup(B,Q)))=n-(|Q|-|Q^A|+1)+1=n-|Q|+|Q^A|$. Therefore, $|A|-|Q^A|+n-|Q|+|Q^A|=n-|Q|+|A|$ as desired.

In the general case $(A,P^A_{\bot})$, with $r(A,P^A_{\bot})=r(\emptyset,P^A_{\bot})+1=|A|$, the sum is $r(A,P^A_{\bot})+r(B,Q)=n-|Q|+1+|A|$. The rank
$r((A,P^A_{\bot})\sqcap(B,Q))$ of the meet equals $|A|-|Q^A|$ if $A\cap B=\emptyset$ and $|A|-|Q^A|+1$ if $A\cap B\neq\emptyset$. Correspondingly, the rank
$r(cl((A,P^A_{\bot})\sqcup(B,Q)))$ of the join equals\\ $n-(|Q|-|Q^A|+1-1)+1=n-|Q|+|Q^A|+1$ if $A\cap B=\emptyset$ and\\ $n-(|Q|-|Q^A|+1)+1=n-|Q|+|Q^A|$ if
$A\cap B\neq\emptyset$.\\
Thus the sought sum is $|A|-|Q^A|+n-|Q|+|Q^A|+1=n-|Q|+1+|A|$ in the former case and $|A|-|Q^A|+1+n-|Q|+|Q^A|=n-|Q|+1+|A|$ in the latter, i.e. the same.

Turning to the \textquoteleft only if\textquoteright\text{ }part, it is not difficult to check that, just like the rank of partitions, the rank of embedded subsets is submodular,
that is to say $r(A,P)+r(B,Q)\geq r(cl((A,P)\sqcup(B,Q)))+r((A,P)\sqcap(B,Q))$, with strict inequality whenever $(A,P),(B,Q)$ is not a modular pair \cite[p. 58]{Aigner79} of
lattice elements.
\end{proof}
\smallskip

Let $N_+=\{1,\ldots ,n,n+1\}$, i.e. $N_+=N\cup\{n+1\}$ and denote by $\mathcal P^{N_+}$ the lattice of partitions of $N_+$, with generic element denoted by
$P^+\in\mathcal P^{N_+}$. From Section 1, $|\mathcal E^N|=\mathcal B_{n+1}=|\mathcal P^{N_+}|$. In fact, there is a lattice isomorphism
$f:\mathcal E^N\rightarrow\mathcal P^{N_+}$, meaning that for all $(A,P),(B,Q)\in\mathcal E^N$ the following holds:
\begin{eqnarray*}
f((A,P)\sqcap(B,Q))&=&f(A,P)\wedge f(B,Q)\text ,\\
f(cl((A,P)\sqcup(B,Q)))&=&f(A,P)\vee f(B,Q)\text .
\end{eqnarray*}
To see this, consider that any partition $\{A_1,\ldots ,A_{|P|}\}=P\in\mathcal P^N$ has associated $|P|+1$ partitions $P^+_1,\ldots ,P^+_{|P|},P^+_{|P|+1}\in\mathcal P^{N_+}$,
all inducing the same partition $P$ of $N$. The first $|P|$ ones $P^+_1,\ldots ,P^+_{|P|}$ obtain each by placing additional element $n+1$ into block $A_k\in P,1\leq k\leq|P|$,
while the last one $P^+_{|P|+1}$ obtains by placing $n+1$ into a singleton block $\{n+1\}\in P^+_{|P|+1}$. Of course, the former are pair-wise incomparable, i.e.
$P^+_k\not\geqslant P^+_{k'}\not\geqslant P^+_k$ for $1\leq k<k'\leq|P|$, while the latter is strictly finer than all the others: $P^+_k>P^+_{|P|+1},1\leq k\leq|P|$. The same
applies to embedded subsets $(A_1,P),\ldots ,(A_{|P|},P),(\emptyset,P)\in\mathcal E^N$, in that $(A_k,P)\not\sqsupseteq(A_{k'},P)\not\sqsupseteq(A_k,P)$ for $1\leq k<k'\leq|P|$
as well as $(A_k,P)\sqsupset(\emptyset,P)$, $1\leq k\leq|P|$. In other terms, letting $P=\{A,B_1,\ldots,B_{|P|-1}\}\in\mathcal P^N$ for all $(A,P)\in\mathcal E^N$ such that
$A\neq\emptyset$, isomorphism $f:\mathcal E^N\rightarrow\mathcal P^{N_+}$ is defined by
\begin{equation*}
f(A,P)=P^+\text{ such that }P^+=\{A\cup\{n+1\},B_1,\ldots,B_{|P|-1}\}\text .
\end{equation*}
On the other hand, for every $P=\{B_1,\ldots ,B_{|P|}\}\in\mathcal P^N$ and $(\emptyset,P)\in\mathcal E^N$,
\begin{equation*}
f(\emptyset,P)=P^+\text{ such that }P^+=\{\{n+1\},B_1,\ldots,B_{|P|}\}\text .
\end{equation*}
Evidently, in this way the $2^{n+1}-(n+1)$ modular elements of $\mathcal P^{N_+}$ bijectively correspond to the $2^n-n+2^n-1=2^{n+1}-(n+1)$ modular elements of $\mathcal E^N$. 

\section{M\"obius inversion}
M\"obius inversion applies to locally finite posets \cite{Rota64}. In particular, for a finite lattice $(L,\wedge,\vee)$ ordered by $\geqslant$, with generic elements
$x,y,z\in L$ and bottom element $x_{\bot}$, any lattice function $f: L\rightarrow\mathbb R$ has M\"obius inversion (from below) $\mu^f:L\rightarrow\mathbb R$ given by
$\mu^f(x)=\sum_{x_{\bot}\leqslant y\leqslant x}\mu_L(y,x)f(y)$, where $\mu_L$ is the M\"obius function of $L$, defined recursively on ordered pairs $(y,x)\in L\times L$ by
$\mu_L(y,x)=1$ if $y=x$ as well as $\mu_L(y,x)=-\sum_{y\leqslant z<x}\mu_L(y,z)$ if $y<x$ (i.e. $y\leqslant x$ and $y\neq x$), while $\mu_L(y,x)=0$ if $y\not\leqslant x$.
Accordingly, $f$ and $\mu^f$ are related by $f(x)=\sum_{y\leqslant x}\mu^f(y)$ for all $x\in L$. The M\"obius function of the subset lattice is
$\mu_{2^N}(B,A)=(-1)^{|A\backslash B|}$, with $B\subseteq A$. Concerning the M\"obius function of $\mathcal P^N$, given any two partitions $P,Q\in\mathcal P^N$, if
$Q<P=\{A_1,\ldots ,A_{|P|}\}$, then for every $A\in P$ there are $B_1,\ldots ,B_{k_A}\in Q$ such that $A=B_1\cup\cdots\cup B_{k_A}$, with $k_A>1$ for at least one $A\in P$.
Segment (or interval) $[Q,P]=\{P':Q\leqslant P'\leqslant P\}$ is isomorphic to product $\times_{A\in P}\mathcal P^{k_A}$, where $\mathcal P^k$ denotes the lattice of partitions
of a $k$-set. Let $c^{QP}_k=|\{A:A\in P,k_A=k\}|,1\leq k\leq n$ be the number of blocks of $P$ obtained as the union of $k$ blocks of $Q$, where
$c^{QP}=(c^{QP}_1,\ldots ,c^{QP}_n)\in\mathbb Z_+^n$ is the class (or type) of segment $[Q,P]$. Then \cite[pp. 359-360]{Rota64},
\begin{equation}
\mu_{\mathcal P^N}(Q,P)=(-1)^{-n+\sum_{1\leq k\leq n}c^{QP}_k}\prod_{1<k<n}(k!)^{c^{QP}_{k+1}}\text .
\end{equation}
This result obtains since the M\"obius function is multiplicative \cite[pp. 147-149]{Aigner79} \cite[p. 560]{Greene82} \cite[Theorem 13.6]{GesselStanley95}
\cite[p. 305]{Stanley2012EnuCom}, which also entails here that for any two pairs $(A,P),(B,Q)\in X^N$ the M\"obius function $\mu_{X^N}$ of $X^N$ is given by
$\mu_{X^N}((B,Q),(A,P))=\mu_{2^N}(B,A)\mu_{\mathcal P^N}(Q,P)$.

Another fundamental result in the present setting concerns the link between the M\"obius function and a closure operator \cite[p. 167]{Aigner79} \cite[pp. 561-562]{Greene82}
\cite[p. 179]{Greene73} \cite[p. 349]{Rota64} \cite[Ex. 84, p. 425]{Stanley2012EnuCom}. For closed pairs $(A,P),(B,Q)\in\mathcal E^N$, and denoting generic pairs by
$(A',P')\in X^N$, the M\"obius function $\mu_{\mathcal E^N}$ of $\mathcal E^N$ is
\begin{equation}
\mu_{\mathcal E^N}((B,Q),(A,P))=\sum_{\underset{cl(A',P')=(A,P)}{(A',P')\in X^N}}\mu_{X^N}((B,Q),(A',P'))\text .
\end{equation}

The M\"obius function $\mu_{\mathcal E^N}$ of $\mathcal E^N$ is now determined through recursion following \cite[pp. 359-360]{Rota64}. Of course, this leads to
re-obtain the isomorphism observed in Section 3. Firstly some further isomorphisms concerning segments $[(B,Q),(A,P)]\subseteq\mathcal E^N$ have to be highlighted.

Extending the notation introduced above, for $0\leq m\leq n$ denote by $\mathcal P^m$ the lattice of partitions of a $m$-set. Also let $\mathcal E^m$ denote the lattice of
embedded subsets of a $m$-set, hence $\mathcal E^0=\{(\emptyset,\emptyset)\}$ consists of only one element. Similarly, recall that the recursions for the Stirling numbers
of the second kind and for the Bell numbers \cite[pp. 89, 92]{Aigner79} rely on the convention that the number of partitions of the empty set is 1, i.e.
$\mathcal P^0=\{\emptyset\}$. Finally, set $\mu_{\mathcal P^m}=\mu_{\mathcal P^m}(\bot,\top)$ and $\mu_{\mathcal E^m}=\mu_{\mathcal E^m}(\bot,\top)$, where $\bot$ and $\top$ are
the bottom and top elements of the corresponding lattice, with $\mu_{\mathcal P^m}=(-1)^{m-1}(m-1)!$ from \cite[Prop. 3 p. 359]{Rota64}, and
$\mu_{\mathcal P^0}=1=\mu_{\mathcal E^0}$ (see \cite[p. 319]{Stanley2012EnuCom}; also, $-1!=-1$ \cite{ConcreteMathematics}). 

Now consider any non-empty segment $[(B,Q),(A,P)]\subseteq\mathcal E^N$, that is to say $(A,P)\sqsupseteq(B,Q)$. Firstly, if $B=A$, then $[(B,Q),(A,P)]\cong[Q^{A^c},P^{A^c}]$,
where $\cong$ denotes isomorphism and $[Q^{A^c},P^{A^c}]\subseteq\mathcal P^{|A^c|}$, hence $\mu_{\mathcal E^N}((B,Q),(A,P))$ is given by expression (2), i.e.
\begin{eqnarray*}
\mu_{\mathcal E^N}((B,Q),(A,P))&=&\mu_{\mathcal P^{|A^c|}}(Q^{A^c},P^{A^c})\\
&=&(-1)^{-|A^c|+\sum_{1\leq k\leq |A^c|}c^{Q^{A^c}P^{A^c}}_k}\prod_{1<k<|A^c|}(k!)^{c^{Q^{A^c}P^{A^c}}_{k+1}}\text ,
\end{eqnarray*}
with $c^{Q^{A^c}P^{A^c}}\in\mathbb Z_+^{|A^c|}$ being the class of segment $[Q^{A^c},P^{A^c}]$.

Secondly, consider the case where $\emptyset\neq B\subset A$. Letting $B=B^A_1$, partition $Q$ has form $Q=\{B_1^A,\ldots ,B_{|Q^A|}^A\}\cup Q^{A^c}$, where
$\{B_1^A,\ldots ,B_{|Q^A|}^A\}=Q^A$. It is readily recognized
that for any atom $(i,P_{\bot})\in\mathcal A_1$ of $\mathcal E^N$ (see Section 2), segment $[(i,P_{\bot}),(N,P^{\top})]\subset\mathcal E^N$ is isomorphic to
$\mathcal P^N$ (see also \cite[Prop. 2 (viii), p. 482]{Grabisch2010}). Analogously, denoting by $\{A\}$ the top element of $\mathcal P^{|A|}$, segment
$[(B,Q^A),(A,\{A\})]\subset\mathcal E^{|Q^A|}$ is isomorphic to $\mathcal P^{|A|}$. Therefore,
\begin{eqnarray*}
\mu_{\mathcal E^N}((B,Q),(A,P))&=&\mu_{\mathcal P^{|Q^A|}}\mu_{\mathcal P^{|A^c|}}(Q^{A^c},P^{A^c})\\
&=&(-1)^{|Q^A|-1}(|Q^A|-1)!\mu_{\mathcal P^{|A^c|}}(Q^{A^c},P^{A^c})\text .
\end{eqnarray*}

The third and most relevant case is when $\emptyset=B\subset A$. Like in the previous case, partition $Q$ has form $Q^A\cup Q^{A^c}$, and thus
$[(\emptyset,Q^A),(A,\{A\})]\cong\mathcal E^{|Q^A|}$. Accordingly, $[(\emptyset,Q),(A,P)]\cong\mathcal E^{|Q^A|}\times[Q^{A^c},P^{A^c}]$, entailing
\begin{equation*}
\mu_{\mathcal E^N}((B,Q),(A,P))=\mu_{\mathcal E^{|Q^A|}}\mu_{\mathcal P^{|A^c|}}(Q^{A^c},P^{A^c})\text .
\end{equation*}
Hence the M\"obius function $\mu_{\mathcal E^N}:\mathcal E^N\times\mathcal E^N\rightarrow\mathbb Z$ is fully determined once
$\mu_{\mathcal E^N}((\emptyset,P_{\bot}),(N,P^{\top}))=\mu_{\mathcal E^n}$ is known. To this end, focus on the dual atoms of $\mathcal E^N$, i.e. closed pairs
covered by the top element $(N,P^{\top})$. There are $2^n-2$ such dual atoms with form $(A,\{A,A^c\})$ for $\emptyset\neq A\neq N$, plus a further one with form
$(\emptyset,P^{\top})$, where $2^n-1$ is precisely the number of dual atoms of $\mathcal P^{n+1}$. Also recall that from Weisner's theorem \cite[Th. 13.9]{GesselStanley95}
\cite[Cor. (b), p. 351]{Rota64} \cite[Cor. 3.9.3, p. 314]{Stanley2012EnuCom}, for any $(A,P),(B,Q)\in\mathcal E^N$ with $(N,P^{\top})\sqsupset(A,B)$
\begin{equation}
\sum_{\underset{(A',P')\sqcap(A,P)=(B,Q)}{(A',P')\in\mathcal E^N}}\mu_{\mathcal E^N}((A',P'),(N,P^{\top}))=0\text .
\end{equation}
\begin{proposition}
The M\"obius function $\mu_{\mathcal E^n}$ of the lattice of embedded subsets of a $n$-set is $\mu_{\mathcal E^n}=\mu_{\mathcal P^{n+1}}=(-1)^nn!$.
\end{proposition}

\begin{proof}
In expression (4), let $(B,Q)=(\emptyset,P_{\bot})$ and choose $(A,P)$ to be a dual atom of the form $(N\backslash i,\{N\backslash i,i\})$ for $i\in N$. Then,
$(A',P')\sqcap(A,P)=(\emptyset,P_{\bot})$ entails that $(A',P')\in\mathcal A$ is an atom and, in particular, either $(A',P')\in\mathcal A_2$ has form $(\emptyset,[ij])$ with
$j\in N\backslash i$, or else $(A',P')\in\mathcal A_1$ has form $(i,P_{\bot})$. Accordingly,
\begin{equation*}
\mu_{\mathcal E^n}=-\sum_{\underset{(A',P')\sqcap(A,P)=(\emptyset,P_{\bot})}{(A',P')\in\mathcal A}}\mu_{\mathcal E^N}((A',P'),(N,P^{\top}))\text ,
\end{equation*}
where the sum ranges over the $n$ atoms just described. From the above isomporphisms, $[(i,P_{\bot}),(N,P^{\top})]\cong\mathcal P^n$ as well as
$[(\emptyset,[ij]),(N,P^{\top})]\cong\mathcal E^{n-1}$. Thus $\mu_{\mathcal E^n}=-(\mu_{\mathcal P^n}+(n-1)\mu_{\mathcal E^{n-1}})$. Since
$\mu_{\mathcal E^0}=\mu_{\mathcal P^0}=1$ and $\mu_{\mathcal E^1}=\mu_{\mathcal P^2}=-1$, the desired conclusion
$\mu_{\mathcal E^n}=-(\mu_{\mathcal P^n}+(n-1)\mu_{\mathcal P^n})=-n\mu_{\mathcal P^n}=\mu_{\mathcal P^{n+1}}$ follows.
\end{proof}
\smallskip

It may be noted that this route aims at following tightly \cite[pp. 359-360]{Rota64}. However, if $(A,P)$ is (straightforwardly) chosen to be the dual atom
$(\emptyset,P^{\top})$ (rather than $(N\backslash i,\{N\backslash i,i\})$), then the above argument immediately yields $\mu_{\mathcal E^n}=-n\mu_{\mathcal P^n}$,
as the $n$ involved atoms are all the elements of $\mathcal A_1$, i.e. $(j,P_{\bot}),j\in N$. 

Proposition 4 confirms that the lattice $\mathcal E^n$ of embedded subsets of a $n$-set is isomorphic to the lattice $\mathcal P^{n+1}$ of partitions of a $n+1$-set.
In particular, the characteristic polynomial \cite[p. 155]{Aigner79}, \cite[pp. 319-320]{Stanley2012EnuCom}, \cite{Zaslavsky87} is
\begin{equation*}
\chi_{\mathcal E^n}(x)=\chi_{\mathcal P^{n+1}}(x)=(x-1)(x-2)\cdots(x-n)\text ,
\end{equation*}
with the Whitney numbers of the first kind given by the Stirling numbers of the first kind \cite{Aigner87}, \cite[p. 88]{Aigner79}, i.e.
\begin{equation*}
\sum_{\underset{r(A,P)=k}{(A,P)\in\mathcal E^N}}\mu_{\mathcal E^N}((\emptyset,P_{\bot}),(A,P))=s_{n+1,n+1-k}\text{ for }0\leq k\leq n\text ,
\end{equation*}
and the Whitney numbers of $\mathcal E^n$ of the second kind given by the Stirling numbers of the second kind, i.e.
\begin{equation*}
|\{(A,P):(A,P)\in\mathcal E^N,r(A,P)=k\}|=S_{n+1,n+1-k}\text{ for }0\leq k\leq n\text .
\end{equation*}
Furthermore, $\mathcal E^n$ is a supersolvable lattice, since in view of Proposition 3 it admits $\frac{(n+1)!}{2}$ maximal chains from the bottom element to the top one
consisting of modular elements \cite[Ex. 2.6, p. 104]{StanleySupersolvable} \cite[Ex. 3.14.4, p. 339]{Stanley2012EnuCom}.

It seems also worth pointing out (again, see Section 1) that while in the non-atomic lattice $\mathfrak C(N)_{\bot}$ of embedded subsets only the $2^n-1$ modular elements
$(A,P^A_{\bot}),A\neq\emptyset$ admit a decomposition as a join of the $n$ available atoms $(i,P_{\bot})$ (see \cite[Prop. 8, p. 486]{Grabisch2010}), in the geometric
lattice $\mathcal E^N$ all elements admit a decomposition as a join of the $\binom{n+1}{2}$ atoms. Accordingly \cite[Theorem 13.7]{GesselStanley95},
\begin{equation*}
\mu_{\mathcal C(N)_{\bot}}((\emptyset,P_{\bot}),(A,P))=\left\{\begin{array}{c} (-1)^{|A|}\text{ if }P=P^A_{\bot}\\0\text{ otherwise} \end{array}\right.
\end{equation*}
while
\begin{equation*}
\mu_{\mathcal E^N}((\emptyset,P_{\bot}),(A,P))=(-1)^{|A|}|A|!\prod_{B\in P^{A^c}}(-1)^{|B|-1}(|B|-1)!
\end{equation*}
and $\mu_{\mathcal E^N}((\emptyset,P_{\bot}),(A,P))$ equals the difference between the number of subsets
$\{(A',P')_1,\ldots ,(A',P')_{|S|}\}=S\subseteq\mathcal A$ of atoms of $\mathcal E^N$ such that $|S|$ is even as well as $cl((A',P')_1\sqcup\cdots\sqcup(A',P')_{|S|})=(A,P)$ and
the number of subsets $\{(A',P')_1,\ldots ,(A',P')_{|S'|}\}=S'\subseteq\mathcal A$ of atoms such that $|S'|$ is odd as well as
$cl((A',P')_1\sqcup\cdots\sqcup(A',P')_{|S'|})=(A,P)$ \cite[p. 562]{Greene82}.

Finally, it seems worth noticing that
\begin{equation*}
(-1)^{r(N,P^{\top})}\mu_{\mathcal E^N}((\emptyset,P_{\bot}),(N,P^{\top}))=(-1)^n(-1)^nn!=(-1)^{2n}n!>0
\end{equation*}
for all $n\geq 0$, which is consistent with \cite[Cor. 13.10]{GesselStanley95}.

\section{M\"obius algebra, vector subspaces and games}
The purpose of this section is to highlight that all cooperative games are in fact elements of subspaces of the free vector space $V(\mathcal E^N)$ over $\mathbb R$ generated by
$\mathcal E^N$ \cite[p. 143]{Aigner87}, \cite[p. 181]{Aigner79}. That is, the elements of $V(\mathcal E^N)\subset\mathbb R^{\mathcal B_{n+1}}$ are real-valued lattice functions
$g:\mathcal E^N\rightarrow\mathbb R$.

Recall that, for $N$ regarded as a player set, cooperative game theory deals with three types of lattice functions:
\begin{enumerate}
\item[$(a)$] coalitional games \cite{Roth88} or set functions $v:2^N\rightarrow\mathbb R_+,v(\emptyset)=0$,
\item[$(b)$] global games \cite{GilboaLehrer90GG} or partition functions $h:\mathcal P^N\rightarrow\mathbb R_+,h(P_{\bot})=0$,
\item[$(c)$] games in partition function form \cite{Grabisch2010,Myerson88,ThrallLucas63}, hereafter referred to as PFF games for short, or functions
$f:\mathfrak C(N)_{\bot}\rightarrow\mathbb R_+,f(\emptyset,P_{\bot})=0$.
\end{enumerate}
A further fourth type of cooperative games may be defined, which incorporates all previous ones $(a)$, $(b)$ and $(c)$, namely
\begin{enumerate}
\item[$(d)$] extended PFF games, or those elements $g\in V(\mathcal E^N)$ of the free vector space under concern such that
$g:\mathcal E^N\rightarrow\mathbb R_+,g(\emptyset,P_{\bot})=0$. 
\end{enumerate}

Together with the M\"obius function, the other element of the incidence algebra (of locally finite posets \cite[p. 138]{Aigner79} \cite[p. 344]{Rota64}) playing a fundamental role
in cooperative game theory is the zeta function $\zeta$, as it provides the so-called \textquoteleft unanimity games\textquoteright\text{ }\cite{Roth88}. In the incidence
algebra of $\mathcal E^N$ over $\mathbb R$, the zeta function $\zeta_{\mathcal E^N}:\mathcal E^N\times\mathcal E^N\rightarrow\{0,1\}$ is defined on ordered pairs of embedded
subsets by
\begin{equation*}
\zeta_{\mathcal E^N}((B,Q),(A,P))=\left\{\begin{array}{c} 1\text{ if }(A,P)\sqsupseteq(B,Q)\text ,\\0\text{ otherwise.} \end{array}\right.
\end{equation*}
In these terms, unanimity (coalitional) games $u_B,\emptyset\neq B\in 2^N$ are defined by $u_B(A)=\zeta_{2^N}(B,A)$ for all subsets (or coalitions) $A\in 2^N$, where $\zeta_{2^N}$
is the zeta function in the incidence algebra of $2^N$. They are linearly independendent elements of the vector space of coalitional games and thus form a basis. This fact is
essential since the very beginning of cooperative game theory.

In the present setting, for $(B,Q)\in\mathcal E^N$, let $\zeta^{B,Q}(A,P)=\zeta_{\mathcal E^N}((B,Q),(A,P))$ for all $(A,P)\in\mathcal E^N$. Then,
$\{\zeta^{B,Q}:(B,Q)\in\mathcal E^N\}$ is a basis of $V(\mathcal E^N)$. Also denote by $V(\mathcal E^N)_{\bot}\subset\mathbb R^{\mathcal B_{n+1}-1}_+$ the subspace of extended
PFF games. Furthermore, let the M\"obius inversion of $g$ be $\mu^g:\mathcal E^N\rightarrow\mathbb R$, i.e.
\begin{eqnarray*}
\mu^g(A,P)&=&\sum_{\underset{(B,Q)\in\mathcal E^N}{(B,Q)\sqsubseteq(A,P)}}\mu_{\mathcal E^N}((B,Q),(A,P))g(B,Q)=\\
&=&g(A,P)-\sum_{\underset{(B,Q)\in\mathcal E^N}{(B,Q)\sqsubset(A,P)}}\mu^g(B,Q)\text ,
\end{eqnarray*}
with $\mu^g(\emptyset,P_{\bot})=g(\emptyset,P_{\bot})$ ($=0$ if $g\in V(\mathcal E^N)_{\bot}$), and where
\begin{equation*}
\mu^{\zeta^{B,Q}}(A,P)=\left\{\begin{array}{c} 1\text{ if }(A,P)=(B,Q)\\0\text{ otherwise} \end{array}\right.
\end{equation*}
for every element $\zeta^{B,Q}$ of the chosen basis, thereby providing the canonical basis (or Kronecker delta) elements (see \cite[Proposition 4.50, p. 181]{Aigner79}, while
crucially noticing that the alternative definition $\zeta^{B,Q}(A,P)=\zeta_{\mathcal E^N}((A,P),(B,Q))$ applies there.) Specifically, let $\epsilon_{B,Q}\in V(\mathcal E^N)$ be the
$\mathcal B_{n+1}$-vector with entries indexed by elements $(A,P)\in\mathcal E^N$ and such that its unique non-zero entry is the $(B,Q)$ one, which equals 1. Then the M\"obius
algebra of $\mathcal E^N$ over $\mathbb R$ (denoted by M\"ob$(\mathcal E^N)$, see \cite[p. 143]{Aigner87} \cite[p. 184]{Aigner79}) is defined by the system of orthogonal idempotents
$\{\epsilon_{B,Q}:(B,Q)\in\mathcal E^N\}$ and linear extension to all of $V(\mathcal E^N)$ through multiplication
\begin{equation*}
\epsilon_{B,Q}\cdot\epsilon_{A,P}=\left\{\begin{array}{c} 1\text{ if }(A,P)=(B,Q)\text ,\\0\text{ otherwise.} \end{array}\right.
\end{equation*}
In terms of the above basis $\{\zeta^{B,Q}:(B,Q)\in\mathcal E^N\}$ of $V(\mathcal E^N)$, this multiplication yields $\zeta^{B,Q}\cdot\zeta^{A,P}=\zeta^{cl((A,P)\sqcup(B,Q))}$
(see \cite[Prop. 4.55, p. 184]{Aigner79}, crucially noticing again that since the alternative definition of the basis detailed above applies there, then 
$\zeta^{B,Q}\cdot\zeta^{A,P}=\zeta^{(A,P)\sqcap(B,Q)}$). It is well known that M\"obius inversion $\mu^g$ quantifies precisely the coefficients identifying functions
$g\in V(\mathcal E^N)$ as linear combinations of the basis elements, that is
\begin{equation}
g(A,P)=\sum_{(B,Q)\in\mathcal E^N}\mu^g(B,Q)\zeta^{B,Q}(A,P)\text{ for all }(A,P)\in\mathcal E^N\text .
\end{equation}

If $Y\subset\mathcal E^N$ and $\mu^g(A,P)=0$ for all $(A,P)\in\mathcal E^N\backslash Y$, then M\"obius inversion $\mu^g$ may be said to live only $Y$. In these terms, the following
observations are immediate:
\begin{enumerate}
\item[$(i)$] (traditional) PFF games are elements of the subspace $\hat V(\mathcal E^N)\subset\mathbb R^{\mathcal B_{n+1}-\mathcal B_n}$ of $V(\mathcal E^N)$
consisting of those extended PFF games $g\in V(\mathcal E^N)$ whose M\"obius inversion lives only on embedded subsets $(A,P)$ such that $A\neq\emptyset$, i.e.
$A=\emptyset\Rightarrow\mu^g(A,P)=0$, entailing $g(\emptyset,P)=0$ for all $P\in\mathcal P^N$;
\item[$(ii)$] global games are elements of the subspace $\tilde V(\mathcal E^N)\subset\mathbb R^{\mathcal B_n-1}$ consisting of extended PFF games $g\in V(\mathcal E^N)$
whose M\"obius inversion lives only on embedded subsets $(A,P)$ such that $A=\emptyset$ and $P\neq P_{\bot}$, entailing $g(A,P)=g(\emptyset,P)$ for all $A\in P$ (thereby somehow
strengthening the idea of a global level of satisfaction, common to all players, see \cite{GilboaLehrer90GG});
\item[$(iii)$] coalitional games are elements of the subspace $\bar V(\mathcal E^N)\subset\mathbb R^{2^n-1}$ consisting of those $g\in V(\mathcal E^N)$ whose M\"obius
inversion lives only on the $2^n-1$ modular elements $(A,P^A_{\bot})$ of $\mathcal E^N$ defined in Section 3 such that $A\neq\emptyset$, yielding $g(\emptyset,P)=0$ for all
$P\in\mathcal P^N$ as well as $g(A,P)=g(A,P^A_{\bot})$ for all $P\in\mathcal P^N$ such that $A\in P$.
\end{enumerate}

Another subspace of $V(\mathcal E^N)$ appearing in cooperative game theory is given by the so-called ``additively separable'' global games (or partition functions, see
\cite{GilboaLehrer90GG,GilboaLehrer91VI}), namely those $h:\mathcal P^N\rightarrow\mathbb R_+$ such that $h(P)=h_v(P):=\sum_{A\in P}v(A)$ for some coalitional game (or set function)
$v:2^N\rightarrow\mathbb R_+$. The M\"obius inversion of these partition functions lives only on the $2^n-n$ modular elements $P^A_{\bot}$ of the partition lattice, see Section 1.
This is detailed hereafter (see also \cite[Prop. 4.4 and Appx]{GilboaLehrer90GG}, \cite[Prop. 3.3]{GilboaLehrer91VI} and \cite[Prop. 7]{RossiLNCSNearBoolean}). 
\begin{proposition}
If $h=h_v$ and $\sum_{i\in N}v(\{i\})=h_v(P_{\bot})>0$, then a continuuum of set functions $w:2^N\rightarrow\mathbb R_+,w\neq v$ satisfies $h=h_w$.
\end{proposition}
\begin{proof}
By direct substitution, the M\"obius inversion $\mu^{h_v}:\mathcal P^N\rightarrow\mathbb R$ of partition function $h_v$ satisfies
\begin{equation*}
\mu^{h_v}(P)=\sum_{A\in P}\sum_{B\subseteq A}v(B)\sum_{Q\leqslant P:B\in Q}\mu_{\mathcal P^N}(Q,P)\text{ for all }P\in\mathcal P^N\text .
\end{equation*}
If $P\neq P^A_{\bot}$, then the recursion for M\"obius function $\mu_{\mathcal P^N}:\mathcal P^N\times\mathcal P^N\rightarrow\mathbb R$ yields
\begin{equation*}
\sum_{Q\leqslant P:A\in Q}\mu_{\mathcal P^N}(Q,P)=\sum_{P^A_{\bot}\leqslant Q\leqslant P}\mu_{\mathcal P^N}(Q,P)=0\text , 
\end{equation*}
and the same for proper subsets $B\subset A$. M\"obius inversion $\mu^{h_v}$ thus lives only on modular partitions, where it obtains recursively by
$\mu^{h_v}(P_{\bot})=\sum_{i\in N}v(\{i\})$ and $\mu^{h_v}(P_{\bot}^A)=\mu^v(A)$ for $1<|A|\leq n$, with $P^N_{\bot}=P^{\top}$. Hence any $w\neq v$ satisfying
$\sum_{i\in N}v(\{i\})=\sum_{i\in N}w(\{i\})$ and $\mu^v(A)=\mu^w(A)$ for all $A\in 2^N$ such that $|A|>1$ also additively separates $h$, i.e. $h_v=h_w$.
\end{proof}
\smallskip

If $\sum_{i\in N}v(\{i\})>0$, then all $w:2^N\rightarrow\mathbb R_+$ such that $h_v=h_w$ may be chosen from within a $n-1$-dimensional simplex $\Delta\subset\mathbb R^n_+$.
However, if $\sum_{i\in N}v(\{i\})=0$, then there is no $w:2^N\rightarrow\mathbb R_+$ such that $h_v=h_w,v\neq w$. Nevertheless, this technical distinction plays no role as soon as
attention is placed on generic partition and set functions $h:\mathcal P^N\rightarrow\mathbb R,w:2^N\rightarrow\mathbb R$, i.e. not required to take only non-negative real values.
In any case, in view of observation $(ii)$ above, it is evident that additively separable partition functions may be regarded as elements of the subspace of $V(\mathcal E^N)$
consisting of those extended PFF games $g$ with M\"obius inversion $\mu^g$ such that $\mu^g(B,Q)\neq 0$ only if $B=\emptyset$ and $Q$ is a modular partition. Developing from
additively separable partition functions, now consider those $g\in V(\mathcal E^N)$ with M\"obius inversion $\mu^g$ living only on the $2^{n+1}-(n+1)$ modular elements of
$\mathcal E^N$ (see Section 3).
\begin{proposition}
If $g\in V(\mathcal E^N)$ has M\"obius inversion $\mu^g$ living only on modular elements of $\mathcal E^N$, then there are set functions $v,w:2^N\rightarrow\mathbb R,v(\emptyset)=0$
such that
\begin{equation*}
g(A,P)=v(A)+\sum_{B\in P}w(B)\text{ for all }(A,P)\in\mathcal E^N\text .
\end{equation*}
\end{proposition}
\begin{proof}
Under the above conditions, and given expression (5) above,
\begin{equation*}
g(A,P)=\mu^g(\emptyset,P_{\bot})+\sum_{B\in P}\sum_{\emptyset\neq B'\subseteq B}\mu^g(\emptyset,P^{B'}_{\bot})+\sum_{\emptyset\neq A'\subseteq A}\mu^g(A',P^{A'}_{\bot})
\end{equation*}
for all $(A,P)\in\mathcal E^N$. In view of Proposition 5,
\begin{equation*}
\mu^g(\emptyset,P_{\bot})+\sum_{B\in P}\sum_{\emptyset\neq B'\subseteq B}\mu^g(\emptyset,P^{B'}_{\bot})=\sum_{B\in P}w(B)
\end{equation*}
for any set function $w$ with M\"obius inversion $\mu^w$ such that
\begin{equation*}
\mu^g(\emptyset,P_{\bot})=\sum_{i\in N}w(\{i\})\text{ and }\mu^g(\emptyset,P^{B'}_{\bot})=\mu^w(B')\text{ for }B'\in 2^N,|B'|>1\text .
\end{equation*}
On the other hand, $\sum_{\emptyset\neq A'\subseteq A}\mu^g(A',P^{A'}_{\bot})=v(A)=\sum_{A'\subseteq A}\mu^v(A')$ for any set function $v$ with M\"obius inversion $\mu^v$
satisfying $\mu^g(A',P^{A'}_{\bot})=\mu^v(A')$ for $A'\in 2^N$, $|A'|>0$ and $\mu^v(\emptyset)=v(\emptyset)=0$.
\end{proof}
\smallskip

Like for additively separable partition functions, if the M\"obius inversion of extended PFF games lives only on the modular elements of the lattice, then all the values taken on
embedded subsets can be recovered from only two set functions, out of which one is unique while the other can be chosen from within a continuum (in view of Proposition 5).

\section{Concluding remarks}
This paper essentially proposes to look at the so-called lattice of embedded subsets from a wider perspective, such that the novel resulting structure is a combinatorial geometry
obtained through a closure operator that satisfies the Steinitz exchange axiom. With respect to the lattice previously proposed in \cite{Grabisch2010}, the geometric one presented
here additionally includes all elements given by a (non-bottom) partition paired with the empty set, and this evidently yields a most natural lattice embedding, i.e. of $\mathcal P^N$
into $\mathcal E^N$. However, from this perspective what seems mostly important is that the geometric lattice of embedded subsets of a $n$ set is isomorphic to the lattice of partitions
of a $n+1$-set. In general, this enables to re-obtain a variety of results applying to partitions, ranging from the characteristic polynomial to supersolvability.

Since embedded subsets, or embedded coalitions, were firstly used as modeling tools in cooperative game theory, a meaningful direction for investigation focuses on the free vector
space of functions taking real values on lattice elements. In this view, a novel type of games has been defined, namely extended PFF games, which englobe all existing cooperative games
as elements of suitable vector subspaces. But perhaps most importantly, the geometric lattice of embedded subsets yields that all existing cooperative games are real-valued functions
defined on atomic lattices, and this may have fundamental implications for the solution concept (i.e. how to share the fruits of cooperation) associated with these games. This shall
be addressed in future work. 

\bibliographystyle{abbrv}
\bibliography{biblioContinuousSetPacking}

\end{document}